\documentclass[a4paper,11pt]{article}

\usepackage{amsmath}
\usepackage{amssymb}
\usepackage{amsthm}
\usepackage[utf8]{inputenc}
\usepackage[backend=bibtex,style=numeric,giveninits=true,maxbibnames=99]{biblatex}
\usepackage[labelfont=bf,labelsep=period]{caption}
\usepackage[nodayofweek]{datetime}
\usepackage{enumitem}
\usepackage{float}
\usepackage[margin=2.5cm]{geometry}
\usepackage{graphicx}
\usepackage{hyperref}
\usepackage{numbertabbing} % Attention, this is a custom style!
\usepackage{thm-restate}
\usepackage{newtx}
\usepackage{url}
\usepackage{xcolor}
\usepackage{xspace}

\graphicspath{{figures/}}

\usepackage{booktabs}
\usepackage{cleveref}
\usepackage{fontawesome5}
\usepackage{subcaption}
\usepackage{tikz}
\usetikzlibrary{math}

\usepackage{mathtools}
\DeclarePairedDelimiter{\abs}{\lvert}{\rvert}
\DeclarePairedDelimiter{\ceil}{\lceil}{\rceil}
\DeclarePairedDelimiter{\floor}{\lfloor}{\rfloor}

\DeclareMathOperator{\fullf}{full}
\DeclareMathOperator{\partialf}{partial}

\renewbibmacro*{doi+eprint+url}{%
  \iftoggle{bbx:url}
  {\iffieldundef{doi}{\usebibmacro{url}}{}}
  {}%
  \newunit\newblock
  \iftoggle{bbx:eprint}
  {\usebibmacro{eprint}}
  {}%
  \newunit\newblock
  \iftoggle{bbx:doi}
  {\printfield{doi}}
  {}
}

\allowdisplaybreaks[2] % allow some broken displays

\newtheoremstyle{plain-boldhead}% name
  {\topsep}%      Space above
  {\topsep}%      Space below
  {\itshape}%     Body font
  {}%         Indent amount (empty = no indent, \parindent = para indent)
  {\bfseries}% Thm head font
  {.}%        Punctuation after thm head
  { }%     Space after thm head: " " = normal space; \newline = linebreak
  {\thmname{#1}\thmnumber{ #2}\thmnote{ (\bfseries #3)}}%    Thm head spec
\newtheoremstyle{definition-boldhead}% name
  {\topsep}%      Space above
  {\topsep}%      Space below
  {\normalfont}% Body font
  {}%         Indent amount (empty = no indent, \parindent = para indent)
  {\bfseries}% Thm head font
  {.}%        Punctuation after thm head
  { }%     Space after thm head: " " = normal space; \newline = linebreak
  {\thmname{#1}\thmnumber{ #2}\thmnote{ (\bfseries #3)}}%    Thm head spec
\theoremstyle{plain-boldhead}
\newtheorem{theorem}{Theorem}

\newtheorem{lemma}[theorem]{Lemma}

\theoremstyle{definition-boldhead}
\newtheorem{definition}{Definition}

\newdateformat{simple}{\THEDAY\ \monthname[\THEMONTH]\ \THEYEAR}
\simple

\floatstyle{ruled}
\newfloat{algo}{htbp}{algo}
\floatname{algo}{Algorithm}

\def \ifempty#1{\def\temp{#1} \ifx\temp\empty }

\newcommand{\var}[1]{\textit{#1}}

\newcommand{\etal}{\emph{et al.}}

\newcommand{\BF}{\ensuremath{\mathbb{F}}\xspace}

\newcommand{\BQ}{\ensuremath{\mathbb{Q}}\xspace}

\newcommand{\CA}{\ensuremath{\mathcal{A}}\xspace}

\newcommand{\CF}{\ensuremath{\mathcal{F}}\xspace}

\newcommand{\CP}{\ensuremath{\mathcal{P}}\xspace}
\newcommand{\CQ}{\ensuremath{\mathcal{Q}}\xspace}

\newcommand{\Setcard}[1]{\ensuremath{\abs[\Big]{#1}}}

\newcommand{\overbar}[1]{\mkern 1.5mu\overline{\mkern-1.5mu#1\mkern-1.5mu}\mkern 1.5mu}

\newcommand{\varffattrs}[1]{\ensuremath{\var{f}_{#1}}}
\newcommand{\varffattrsbar}[1]{\ensuremath{\overbar{\var{f}_{#1}}}} % Version with 'bar'

\newcommand{\varpfattrs}[1]{\ensuremath{\var{p}_{#1}}}

\newcommand{\varpfprocs}[1]{\ensuremath{\alpha_{#1}}}
\newcommand{\qkconsistent}[1]{$Q^#1$-\hspace{0pt}consistent}
\newcommand{\bkconsistent}[1]{$B^#1$-\hspace{0pt}consistent}
\newcommand{\qkresilient}[1]{$Q^#1$-\hspace{0pt}resilient}
\newcommand{\bkresilient}[1]{$B^#1$-\hspace{0pt}resilient}

\newcommand{\bkconsistency}[1]{$B^#1$ consistency}
\newcommand{\qkresilience}[1]{$Q^#1$ resilience}
\newcommand{\bkresilience}[1]{$B^#1$ resilience}

\newcommand{\rUniverse}[1]{
  {{% we make the whole thing an ordinary symbol
  \kern-\nulldelimiterspace % automatically resize the bar with \right
    \mathcal{P} % the function
    {\upharpoonright_{#1}} % this is the delimiter, wrapping it in curly braces fixes the horizontal spacing.
  }}
}

\begin{document}

\title{\bf Asymmetric Grid Quorum Systems for Heterogeneous Processes}

\author{%
  Michael Senn\footnote{%
    Institute of Computer Science, University of Bern,
    Neubrückstrasse 10, 3012 CH-Bern, Switzerland.
  }\\
  University of Bern\\
  \url{michael.senn@unibe.ch}
  \and Christian Cachin\footnotemark[1]\\
  University of Bern\\
  \url{christian.cachin@unibe.ch}
}

\date{\today}

\maketitle

\begin{abstract}\noindent
  Quorum systems are a common way to formalize failure assumptions in
  distributed systems. Traditionally, these assumptions are shared by all
  involved processes. More recently, systems have emerged which allow
  processes some freedom in choosing their own, subjective or asymmetric,
  failure assumptions. For such a system to work, individual processes'
  assumptions must be compatible.
  However, this leads to a Catch-22-style scenario: How can processes
  collaborate to agree on compatible failure assumptions when they have no
  compatible failure assumptions to start with?

  We introduce \emph{asymmetric grid quorum systems} that allow a group of
  processes to specify heterogeneous trust assumptions independently of each
  other and without coordination. They are based on qualitative attributes
  describing how the processes differ. Each process may select a quorum system
  from this class that aligns best with its subjective view. The available
  choices are designed to be compatible by definition, thereby breaking the
  cycling dependency. Asymmetric grid quorum systems have many applications
  that range from cloud platforms to blockchain networks.
\end{abstract}

\section{Introduction}
\label{sec:introduction}

In resilient distributed systems that use replication, independent processes
jointly perform a task, which can vary from reading and writing data in
storage systems and coordinating computations in the cloud to running smart
contracts on a blockchain.
In order to reason about the behavior of such a system during an execution,
one must first formalize assumptions on the reliability of the processes.  It
is common to foresee that processes may crash and stay silent, that they crash
and recover but suffer from amnesia, or that they are outright malicious and
behave arbitrarily.  This work addresses the last case and considers
\emph{Byzantine} processes.

A wide-spread way to describe which processes may fail at once is by assuming
bounds on their numbers with so-called \emph{threshold failure assumptions},
which capture that less than a fraction of them, such as one-third, will be
jointly faulty.  The more general approach of \emph{(Byzantine) quorum
systems}~\cite{dblp:journals/joc/HirtM00,dblp:journals/dc/MalkhiR98,dblp:journals/siamcomp/NaorW98}
encodes a trust assumption in a set system that specifies collections of
processes that may jointly fail.  Quorum systems are the core abstraction that
form the basis of many protocols.

Most of the academic work and many distributed systems deployed in practice
consider this assumption to be a \emph{global}, or \emph{symmetric}, property,
that is shared by all processes.  For applications to secure and trustworthy
systems, however, this may be relaxed, in the same way as trust among humans
is subjective because individuals trust others based on personal beliefs and
past experience.

Indeed there are blockchain networks such as the XRP Ledger
(\url{https://xprl.org}) and
Stellar~\cite{dblp:conf/sosp/LokhavaLMHBGJMM19,mazieresStellarConsensusProtocol2016}
(\url{https://stellar.org}) which allow their nodes some freedom in choosing
their individual trust assumptions.  At the time of writing (August 2025),
these are the third- and eighth-largest cryptocurrencies by market
capitalization, respectively.
Similarly there is work which provides theoretical foundations for
\emph{asymmetric
trust}~\cite{dblp:journals/dc/AlposCTZ24,dblp:conf/asiacrypt/DamgardDFN07,DBLP:conf/wdag/LiCL23}
and \emph{flexible
trust}~\cite{dblp:conf/ccs/MalkhiN019,DBLP:conf/sp/NeuSYT24}.  These models
generalize Byzantine quorum systems and allow processes to pick quorums of
their own. There already exist diverse protocols for this setting, such as
broadcast
abstractions~\cite{dblp:journals/dc/AlposCTZ24,DBLP:conf/wdag/LiCL23}, a common
coin~\cite[Section 7]{dblp:journals/dc/AlposCTZ24} or a DAG-based
consensus~\cite{dblp:conf/podc/Amores-SesarCVZ25}.

However, asymmetric trust poses a new challenge. If processes are given the
freedom to choose their trust assumptions, these must be compatible with each
other and the processes should have enough mutual trust so that system can act
coherently.
This is no different from the traditional symmetric case, where one requires,
for example, the overlap of any two quorums to have a certain size. But now we
require processes to agree and compromise on their subjective trust
assumption, so that our distributed system can do some work. This is a
distributed task by itself that requires coordination.

Resilient and highly secure computing systems often exploit diversity of
software, hardware, governance, and other aspects to reduce the
undesirable vulnerabilities inherent in a
monoculture~\cite{dblp:journals/misq/ChenKK11,LalaSchneider09}.
% \footnote{\url{https://en.wikipedia.org/wiki/Monoculture_(computer_science)}}.
The processes that comprise such systems are characterized by attributes such
as operating system, CPU model, geographical location, programming language,
virtual machine implementation, jurisdiction and so on, and each attribute can
take on different values.  But designing trust assumptions from such attributes
has been done ad hoc, so far.

This paper introduces a new approach to designing trust models for distributed
systems: We introduce \emph{asymmetric grid quorum systems} that characterize
the components or processes in a system with a number of attributes taken from
a predefined set.  If such attributes offer sufficient diversity, i.e., take
on enough different values, then each process may subjectively choose the
attributes that matter to itself and define its own quorum system without
coordination in the system.  The resulting choices are compatible with each
other by construction~--- in other words, the assumptions made individually
according to the attribute structure automatically lead to an asymmetric
Byzantine quorum system.

In particular, we allow processes to choose their trust assumption from a
fixed set of options characterized by attributes.  We call this a
\emph{structured failure assumption} and formalize the corresponding
fail-prone systems.  The attributes and their values are common knowledge, but
each member may then select which attributes are relevant for itself.  The
attributes span a multi-dimensional grid whose properties we analyze in this
work.  We model fail-prone systems with a grid structure and characterize
their necessary and sufficient properties that lead to asymmetric grid quorum
systems.

We remark that real-world processes are rarely homogeneous, hence, this
technique appears to have practical applications.  Our running example will
treat the operating system of a process and its geographical location (or the
cloud provider supplying the physical hardware) as two relevant attributes
that form a grid, but our methods generalize to a hyperrectangle structure of
arbitrary dimension.
Intuitively, we want to allow that each process chooses which of the other
processes' attributes it believes to be the best predictor of future failure.
The process will then be able to tolerate a large number of faults which
correlate with its chosen attribute, plus a smaller number of processes that
are characterized by other attributes.  In particular, the failure scenarios
that the resulting asymmetric quorum system can tolerate are strictly larger
than those tolerated by threshold quorum systems.

\Cref{sec:related_work} presents related work, and \cref{sec:system_model}
introduces the model we operate in. \Cref{sec:structured_quorum_systems} shows
how to construct failure assumptions for heterogeneous processes, which are
compared with commonly-used threshold failure assumptions in
\cref{sec:comparison_with_threshold}. Finally, \cref{sec:conclusion} concludes.

\section{Related work}
\label{sec:related_work}

\paragraph*{Structured symmetric failure assumptions}

Constructing resilient systems from redundant components has long been an
important technique when faults occur
independently~\cite{DBLP:journals/tse/Avizienis85}.
Threshold-based quorum systems are suitable when one believes faults to be
distributed uniformly at random, with no existing structure on the set of
processes. If one is willing to impose some structure, however, one can design
quorum systems to be more resilient or efficient, at the cost of resilience if
the structure is violated. A prominent example are grid-based quorum systems,
where processes are arranged in a two-dimensional grid, and quorums consist of
one full row and one process from every row
below~\cite{dblp:conf/icde/CheungAA90,dblp:journals/tocs/Maekawa85,dblp:journals/dc/PelegW97}.
However, these systems do not address the manner by which processes should be
arranged in such a grid.

Earlier work by Cachin~\cite{DBLP:conf/dsn/Cachin01} suggests to build
generalized quorum systems based on attributes that distinguish processes,
such that the resulting structures are also suitable for distributed
cryptography.

Whittaker \etal~\cite{dblp:conf/eurosys/WhittakerCHHS21} introduce quorum
systems which accommodate the fact that processes are heterogeneous when it
comes to their available resources. Their construction defines quorums which
more evenly spread the read and write load over processes, to optimize
performance. However, they restrict themselves to non-Byzantine faults, i.e. to
settings where there are no actively adversarial processes.

\paragraph*{Practical systems with asymmetric trust}

The Stellar blockchain allows each validator to specify a custom list of other
validators they trust, which is assembled into a hierarchical trust assumption.
The collection of each validator's local trust structure is called a
\emph{federated Byzantine quorum system (FBQS)} and has been formally analyzed
in literature~\cite{dblp:conf/sosp/LokhavaLMHBGJMM19,dblp:conf/wdag/LosaGM19}.
While such an FBQS is no generalization of Byzantine quorum systems, there is
work which adapts existing algorithms to this
model~\cite{dblp:conf/opodis/Garcia-PerezG18}.

Similarly, the XRP Ledger allows each validator to define a \emph{unique node
list (UNL)}. During execution of the consensus protocol, a validator then only
considers votes from other validators contained in its UNL\footnote{\url{https://xrpl.org/docs/concepts/consensus-protocol/unl}}. Formal analysis has
shown that to preserve safety, the UNLs of all validators have to overlap by up
to 90\%, and that guaranteeing liveness requires UNLs to be
identical~\cite{dblp:conf/opodis/Amores-SesarCM20,dblp:journals/corr/abs-1802-07242,dblp:conf/icissp/MauriCD20}. 

Communities and companies associated with both the XRPL and Stellar ecosystems
publish recommendations of quorum sets respectively UNLs for validators to
adopt. While this answers the need of requiring validators to coordinate 
their subjective trust assumptions, the result resembles a centralized, and
symmetrical, system.

\paragraph*{Other forms of asymmetric trust}

Malkhi \etal~\cite{dblp:conf/ccs/MalkhiN019} on \emph{Flexible BFT} introduce
a separate type of \emph{alive-but-corrupt} acceptor processes, which attack
safety but not liveness. Their protocol allows learners, client processes
that obtain the outcome of the BFT protocol, to have subjective
assumptions about the ratio of fully-Byzantine processes versus those which
are alive-but-corrupt.

Li, Chan, and Lesani provide a model for \emph{heterogeneous quorum
systems}~\cite{DBLP:conf/wdag/LiCL23}, where each process has its own set of
quorums. They analyze which conditions must hold so that abstractions like
broadcast are supported, and provide protocols implementing some of these
abstractions.

Heterogeneous Paxos~\cite{arxiv:2011.08253} adapts Byzantine Paxos to allow
learners to have subjective failure assumptions, while also supporting
individual failure assumptions to be a mixture of Byzantine and crash faults.

\textcite{dblp:journals/dc/AlposCTZ24} provide a generalization of
Byzantine quorum systems to the asymmetric setting by extending earlier work of
Damgård \etal~\cite{dblp:conf/asiacrypt/DamgardDFN07}. They also provide
protocols such as broadcast abstractions which operate in their model.

\section{System Model}
\label{sec:system_model}

We assume a permissioned setting with a universe of $n$ independent processes
$\CP = \{1, \ldots, n\}$. These processes have a known identity and are
able to exchange messages with each other via reliable, authenticated
point-to-point links. In any execution, some of these processes are expected to
be faulty. We model these faulty processes using the notion of \emph{subjective
trust} by Alpos \etal~\cite{dblp:journals/dc/AlposCTZ24}. There, each process
has some liberty in choosing its subjective trust assumptions. For context we
first discuss \emph{generalized quorum systems} by Malkhi and
Reiter~\cite{dblp:journals/dc/MalkhiR98}.

\subsection{Symmetric Trust Assumptions}

In the model of symmetric trust, there is a \emph{global} assumption, which all
processes share, about the combinations of processes that may be jointly
faulty. A \emph{failprone set} $F \subset \CP$ is a maximal set of processes
which may jointly be faulty during an execution. The collection of these
failprone sets, the \emph{failprone system} $\CF = \{F_1, \ldots, F_k\} \subset
2^{\CP}$, specifies this global failure assumption. To illustrate, the common
$n > 3f$ threshold failure assumption corresponds to the failprone system
containing all possible failprone sets of cardinality $\ceil{n/3} - 1$.

While failprone systems formalize failure assumptions, protocols are often
described in terms of \emph{quorum systems}. Those formalize the sets of
processes with which a party must have interacted before proceeding to the next
step of the protocol. 

A common way to derive a quorum system from a failprone system is by the
so-called \emph{canonical quorum system}. As will be shown later, the canonical
quorum system describes a convenient relation between failprone systems and
quorum systems where properties of one imply properties of the other.

\begin{definition}[Canonical quorum system]
  \label{def:canonical_quorum_system}
  Given a failprone system~$\CF$, the \emph{canonical quorum system}~$\CQ$
  of~$\CF$ is defined as the bijective complement~\cite[Theorem
  5.4]{dblp:journals/dc/MalkhiR98}~$\overbar{\CF}$ of~$\CF$, i.e.
  \[
    \CQ = \overbar{\CF} = \{\CP \setminus F : F \in \CF\}
    .
  \]
\end{definition}

\paragraph*{Consistency, availability, and resilience}
In order to reason about protocol properties, a notion of \emph{consistency}
and \emph{availability} of quorum systems is required. Intuitively, consistency
ensures that any two quorums intersect in sufficiently many correct processes,
while availability ensures that at least one quorum consists of exclusively
correct ones. 

While usually all quorum systems share the same availability property, they
often differ in their consistency properties. Malkhi and
Reiter~\cite{dblp:journals/dc/MalkhiR98} distinguish between \emph{masking} and
\emph{dissemination} quorum systems, while Alpos
\etal~\cite{dblp:journals/dc/AlposCTZ24} use the notion of $Q^3$ consistency
for the latter, based on earlier work by Hirt and
Maurer~\cite{dblp:journals/joc/HirtM00}.

\begin{definition}[\qkconsistent{3} quorum system]
  \label{def:q3_quorum_system}
  A \emph{\qkconsistent{3} quorum system} $\CQ$ for a failprone system
  $\CF$ is a set system $\CQ \subset 2^{\CP}$ where the following properties
  hold:
  \begin{description}
    \item[Consistency] The intersection of any two quorums is not a subset of
      any failprone set, that is
      \[
        \forall Q_1, Q_2 \in \CQ, 
        \forall F \in \CF : 
        Q_1 \cap Q_2 \not \subseteq F
        .
      \]
    \item[Availability] For any failprone set, there exists a quorum disjoint
      from it, i.e.
      \[
        \forall F \in \CF : \exists Q \in \CQ : F \cap Q = \varnothing
        .
      \]
  \end{description}
\end{definition}

We remark here that $Q^3$ consistency can be generalized to $Q^k$ consistency
for $k \geq 2$~\cite[Section 1.3]{dblp:conf/asiacrypt/DamgardDFN07}. This
family of consistency notions imposes a sequence of increasingly stricter
conditions on the overlap of individual quorums. To illustrate, a $n > 3f$
threshold failure assumption is \qkconsistent{3}, while \emph{masking quorum
systems}~\cite{dblp:journals/dc/MalkhiR98} are \qkconsistent{4}.

For this work we focus on \qkconsistent{3} quorum systems, which have
been shown to be sufficient for many protocols such as consensus with fast
termination in a crash-fault setting~\cite{dblp:conf/pact/BrasileiroGMR01} or
Byzantine consensus~\cite{dblp:journals/tocs/CastroL02,dblp:conf/podc/YinMRGA19}.

In contrast to the definition above, a quorum system's consistency can also be
formalized purely in regards to its failprone system, called the failprone
system's \emph{resilience}. 

\begin{definition}[$Q^3$ resilience]
  \label{def:q3_resilience}
  A failprone system $\CF \subset 2^{\CP}$ is \emph{\qkresilient{3}} if no union of
  three failprone sets covers the universe, that is
  \[
    \forall F_1, F_2, F_3 \in \CF : F_1 \cup F_2 \cup F_3 \subsetneq \CP
    .
  \]
\end{definition}

One can show that a failprone system $\CF$ is \qkresilient{3} if and only if
its canonical quorum system is \qkconsistent{3}~\cite[Theorem
5.4]{dblp:journals/dc/MalkhiR98}.

\begin{definition}
  We adopt the notation of Alpos \etal~\cite{dblp:journals/dc/AlposCTZ24} to denote by
  $\CF^*$ the set containing all possible subsets of every $F' \in \CF$, i.e.,
  \begin{displaymath}
    \CF^* = \bigcup_{F' \in \CF} 2^{F'} = \{ F : F \subseteq F' \land F' \in \CF \}
    .
  \end{displaymath}
\end{definition}

Intuitively, if there is an execution with faulty processes $F$ and a failprone
system $\CF$, then $F \in \CF^*$ simply means that this specific failure mode
was anticipated by the failure assumption, and the system will be able to
handle it.

\subsection{Asymmetric Trust Assumptions}

We now review the notion of asymmetric trust of Alpos
\etal~\cite{dblp:journals/dc/AlposCTZ24} based on earlier work by Damgård
\etal~\cite{dblp:conf/asiacrypt/DamgardDFN07}. Here, each process $i \in \CP$
chooses its subjective trust assumption through its own failprone system
$\CF_i \subset 2^{\CP}$. The collection of all such failprone systems then
forms the \emph{asymmetric failprone system} $\BF = [\CF_1, \ldots, \CF_n]$.

As in the symmetric case, one can derive a \emph{canonical asymmetric quorum
system} from an asymmetric failprone system.
\begin{definition}[Canonical asymmetric quorum system]
  Given an asymmetric failprone system $\BF$, the \emph{canonical asymmetric quorum system} $\BQ$ of
  $\BF$ is defined as the collection of the bijective complement of every
  $\CF_i \in \BF$, that is
  \[
    \BQ = \overbar{\BF} = [\overbar{\CF_i} : \CF_i \in \BF]
    .
  \]
\end{definition}

\paragraph*{Naive and wise processes}

As each process can have its own trust assumption, it is no longer sensible to
believe that every process' assumption always holds. In any execution with
faulty processes $F \subset \CP$, correct processes are either \emph{wise}, or
\emph{naive}. A process $i \in \CP$ is called \emph{wise} if its subjective
trust assumption holds, that is if $F \in \CF_i^*$. Otherwise, it is called
naive. This is a distinction specific to the asymmetric trust setting.

This is relevant as protocols designed for the asymmetric setting may not be
able to guarantee liveness or safety for all correct processes, but only for
wise ones, or even only for a ``self-sufficient'' subset of wise processes.

\paragraph*{Consistency, availability, and resilience} As in the symmetric
setting, an asymmetric quorum system $\BQ$ must fulfill some properties such
that one can use it in protocols with provable guarantees. This inherently
limits a process' freedom in choosing its own subjective trust assumption, as
its assumption must not only be internally consistent, but also ``compatible''
with those of other processes. Alpos \etal~\cite[Definition
4]{dblp:journals/dc/AlposCTZ24} adopt the generalization of \qkconsistent{3}
quorum systems to the asymmetric setting as originally done by Damgård
\etal~\cite[Theorem 2]{dblp:conf/asiacrypt/DamgardDFN07}, and call it
\bkconsistent{3} asymmetric quorum systems.

\begin{definition}[\bkconsistent{3} asymmetric quorum system]
  \label{def:b3_quorum_system}
  A \emph{\bkconsistent{3} asymmetric quorum system} $\BQ$ for an asymmetric failprone
  system $\BF$ is an array $\BQ = [\CQ_1, \ldots, \CQ_n]$, where $\CQ_i \subset
  2^{\CP}$, and the following properties hold:
  \begin{description}
    \item[Consistency] The intersection of any two quorums of any two processes
      contains at least one process which at least one of the two processes
      believes to be correct, that is
      \[
        \forall i, j \in \CP,
        \forall Q_i \in \CQ_i, \forall Q_j \in \CQ_j, 
        \forall F_{i, j} \in \CF_i^* \cap \CF_j^*:
        Q_i \cap Q_j \not \subseteq F_{i, j}
        .
      \]
    \item[Availability] For any failprone set of any process, there exists a
      quorum for that process disjoint from it, i.e.
      \[
        \forall i \in \CP, 
        \forall F_i \in \CF_i : 
        \exists Q_i \in \CQ_i : F_i \cap Q_i = \varnothing
        .
      \]
  \end{description}
\end{definition}

The availability property is a straightforward generalization of its
counterpart in the symmetric setting. The consistency property is a
generalization in the sense that, if all processes were to pick the same
failprone (and quorum) systems, it would degrade to the symmetric property.

Alpos \etal.~\cite[Definition 5]{dblp:journals/dc/AlposCTZ24} further
introduce $B^3$ resilience of an asymmetric failprone system as a
generalization of $Q^3$ resilience, and show~\cite[Theorem
4]{dblp:journals/dc/AlposCTZ24} that an asymmetric failprone system $\BF$ is
\bkresilient{3} if and only if its canonical asymmetric quorum system is
\bkconsistent{3}.

\begin{definition}[$B^3$ resilience]
  \label{def:b3_resilience}
  An asymmetric failprone system $\BF$ is \emph{\bkresilient{3}} if it holds that
  \begin{align*}
    \forall i, j \in \CP,
    \forall F_i \in \CF_i, \forall F_j \in \CF_j,
    \forall F_{i, j} \in \CF_i^* \cap \CF_j^* : 
    F_i \cup F_j \cup F_{i, j} \subsetneq \CP
    .
  \end{align*}
\end{definition}

We remark that, just like $B^3$ consistency is a generalization of $Q^3$
consistency to the asymmetric setting, other forms of $Q^k$ consistency could
also be adapted to the asymmetric setting. While it has been shown that $B^2$
consistency, a weak form of consistency in the asymmetric setting, actually
degrades to a symmetric failure assumption~\cite{dblp:conf/papoc/SennC25}, the
adaption of e.g. $Q^4$ consistency into the asymmetric setting is an open
problem.

\section{Asymmetric Grid Quorum Systems}
\label{sec:structured_quorum_systems}

In this model of asymmetric trust, coordination is required: Processes must
know of each others' quorum systems before execution of the distributed
protocol can commence, and their choices must be compatible. This is limiting
because learning each others' assumptions is, itself, a distributed problem. We
propose to solve this by letting processes choose among a set of well-defined
options, which are compatible by construction but are still meaningfully
different. Then, no coordination is required: Processes simply select the
quorum system that aligns most with their subjective belief, and can commence
with the protocol.

To this end, we remark that real-world processes are commonly heterogeneous,
i.e. they differ in attributes such as their operating system, geographical
location, or software version. Which of these attributes meaningfully affects a
process' trustworthiness is inherently subjective, and so a good starting
point for our quorum system construction.

\paragraph*{Heterogeneous Processes}
We consider a set of $d$ attributes $A_1, \ldots, A_d$. Each attribute $A_j$ can
take on one of $k_j$ many distinct values in a set $\CA_j = \{a_1,
\ldots, a_{k_j}\}$. We now consider processes which differ in these attributes.
That is, each process $i \in \CP$ will take on exactly one attribute value $a
\in \CA_j$ for every attribute $A_j$.

We specifically limit ourselves to the universe $\CP = \CA_1 \times \ldots
\times \CA_d$ where there is \emph{exactly one process} for every distinct
combination of attribute values. Hence, there are a total of $n = \prod_{j =
1}^d k_j$ processes. These processes can be thought of as a finite
$d$-dimensional grid, i.e., a hyperrectangle of attributes, where processes
make up the vertices. A process' location in the grid depends on the values it
takes on for each of the attributes. An example of such a universe, with
attributes for the operating system (``OS'') and location (``Loc''), is show in
\cref{fig:2d_universe}.

In the remainder we limit ourselves to attributes with a cardinality of at
least four. While our construction works with attributes of a lower
cardinality, such attributes add no expressiveness when it comes to modeling
failures. This is analogous to assuming that there are at least $n \geq
4$ processes when considering an $n > 3f$ threshold failure assumption.

\begin{figure}
  \centering
  \begin{subfigure}[t]{.4\textwidth}
    \includegraphics[width=\textwidth]{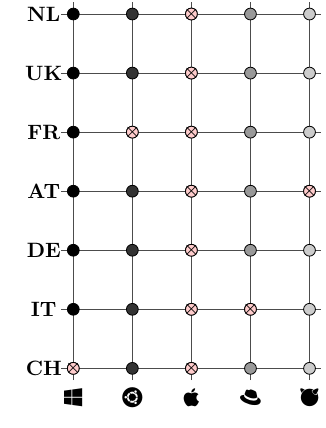}
    \caption{OS belief. The pictured failprone set contains all processes of
      operating system {\faIcon{apple}} and one process for every other operating
    system.}
    \label{fig:2d_universe_os_belief}
  \end{subfigure}\hfill%
  \begin{subfigure}[t]{.4\textwidth}
    \includegraphics[width=\textwidth]{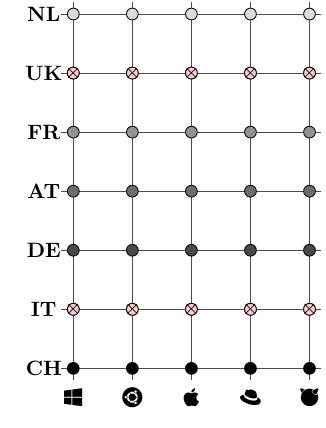}
    \caption{Location belief. The pictured failprone set contains all processes
    for locations IT and UK and zero processes for every other location.}
    \label{fig:2d_universe_loc_belief}
  \end{subfigure}
  \caption{Universe of processes which differ in their operating system
    ({\faIcon{windows}}, {\faIcon{ubuntu}}, {\faIcon{apple}},
    {\faIcon{redhat}}, \faIcon{freebsd})
    and physical location (AT, CH, DE, FR, IT, NL, UK).
    Each grid vertex 
    (\protect\tikz[baseline]{%
      \protect\node[circle,draw,fill=black!80!white,inner sep=0pt,anchor=base,scale=0.8] (A) {$\mathstrut$};
    })
    corresponds to one process. The two
    images each show one example of a set of jointly faulty processes
    (\protect\tikz[baseline]{%
      \protect\node[circle,draw,fill=red!20!white,inner sep=0pt,anchor=base,scale=0.8] (A) {$\times$};
    })
    which can be tolerated by a belief in the respective attribute.
  }
  \label{fig:2d_universe}
\end{figure}

\begin{definition}[Restricted Universe]
  The \emph{restricted universe}\, $\rUniverse{\var{pred}}$ is a subset of the
  universe of processes $\CP$, restricted to those processes which take on
  specific values for certain attributes, as expressed through
  the predicate~$\var{pred}$.
\end{definition}

To illustrate, $\rUniverse{\var{OS} = \text{\faIcon{apple}} \land \var{Loc} \in
\{\var{CH}, \var{IT}\}}$ would be the subset of $\CP$ consisting of all
processes which take on attribute value \faIcon{apple} for attribute
$\var{OS}$, and one of attribute values $\var{CH}$ or $\var{IT}$ for
attribute $\var{Loc}$. 

\begin{lemma}[Cardinality of universe restricted to single attribute value]
  \label{lem:cardinality_of_restricted_universe}
  Consider a subset of $l \leq d$ distinct attributes $\{\hat{A}_1, \ldots,
  \hat{A}_l\} \subset \{A_1, \ldots, A_d\}$ with one attribute value $\hat{a}_j
  \in \hat{\CA}_j$ for each of these attributes. Then the universe restricted
  to processes which take on value $\hat{a}_j$ for attribute $\hat{A}_j$, that
  is $\rUniverse{\hat{A}_1 = \hat{a}_1 \land \ldots \land \hat{A}_l =
  \hat{a}_l}$, forms a $(d-l)$-dimensional subgrid of $\CP$ with cardinality
  $n/ \prod_{j = 1}^l \hat{k}_l$.
\end{lemma}

\begin{proof}
  %!TEX proof = {lem:cardinality_of_restricted_universe}
  By construction, $\CP$ forms a $d$-dimensional grid, with the attributes
  functioning as axis. Limiting $\CP$ to processes taking on one specific value
  $\hat{a}_1$ for $\hat{A}_1$ leads to a $(d-1)$-dimensional subgrid of $\CP$.
  By induction, limiting to processes where attribute $\hat{A}_1 = \hat{a}_1$,
  attribute $\hat{A}_2 = \hat{a}_2$ and so on leads to a $(d - l)$-dimensional
  subgrid. As there is one process in $\CP$ for every distinct combination of
  attribute values, there are a total of $n / \prod_{j = 1}^l \hat{k}_j$
  processes in this subgrid.
\end{proof}

\begin{lemma}[Cardinality of universe restricted to set of attribute values]
  \label{lem:cardinality_of_restricted_universe_set}
  For any two different attributes $A_i \neq A_j$ and any two subsets of their
  values $\CA'_i \subseteq \CA_i$ and $\CA'_j \subseteq \CA_j$ it holds that
  \begin{align*}
    \Setcard{\rUniverse{A_i \in \CA'_i}} = \frac{n}{k_i} \cdot \Setcard{\CA'_i}, \text{ and }
    \Setcard{\rUniverse{A_i \in \CA'_i \land A_j \in \CA'_j}} = \frac{n}{k_i k_j} \cdot \Setcard{\CA'_i} \Setcard{\CA'_j}
    .
  \end{align*}
\end{lemma}

\begin{proof}
  %!TEX proof = {lem:cardinality_of_restricted_universe_set}
  For any attribute values $a_1, a_2 \in \CA'_i$ where $a_1 \neq a_2$, the sets
  $\rUniverse{A_i = a_1}$ and $\rUniverse{A_i = a_2}$ are disjoint. Then, the
  first claim follows by \cref{lem:cardinality_of_restricted_universe}:
  \[
    \Setcard{\rUniverse{A_i \in \CA'_i}}
    = \sum_{a_i \in \CA'_i} \Setcard{\rUniverse{A_i = a_i}}
    = \frac{n}{k_i} \cdot \Setcard{\CA'_i}
    .
  \] 
  The second claim follows by the same argument:
  \[
    \Setcard{\rUniverse{A_i \in \CA'_i \land A_j \in \CA'_j}}
    = \sum_{a_i \in \CA'_i} \sum_{a_j \in \CA'_j} \Setcard{\rUniverse{A_i = a_i \land A_j = a_j}}
    = \frac{n}{k_i k_j} \cdot \Setcard{\CA'_i} \Setcard{\CA'_j}
    .
  \]
\end{proof}

We now define a failprone system $\CF_i$ for each attribute $A_i$. This
failprone system aims to model the belief that attribute $A_i$ is a good
predictor of whether a given process will be faulty, while also supporting some
processes being faulty for unmodeled reasons.

\begin{definition}[Failprone system for $A$-believers]
  \label{def:grid_failprone_system}
  Consider any attribute $A$ with values in $\CA$ where $|\CA| = k$. We define
  the \emph{failprone system $\CF$ for $A$-believers} as the \emph{maximal} set
  system $\CF \subset 2^{\CP}$, where each failprone set $F \in \CF$ is
  characterized by its partition into:
  \begin{itemize}
    \item The \emph{full failures} of $F$, denoted by $\fullf(F)$: For a subset
      of attribute values with cardinality less than $k/3$, all processes which
      take on one of these attribute values.
    \item The \emph{partial failures} of $F$, denoted by $\partialf(F)$: For
      all remaining attribute values, less than $n/(6k)$ of all processes
      \emph{per each of these values}.
  \end{itemize}

  We specifically care about failprone sets that are maximal in terms of their
  cardinalities. More formally, $\CF$ is the collection of all distinct
  failprone sets $F = \fullf(F) \sqcup \partialf(F)$, where for some maximal
  subset of attribute values $\CA' \subset \CA$ with $|\CA'| < |\CA| / 3$,
  \begin{align*}
    \fullf(F)    = \rUniverse{A \in \CA'}, \hspace{1cm} \text{ and } \hspace{1cm}
    \partialf(F) = \bigsqcup_{a \in \CA \setminus \CA'} X_a
    ,
  \end{align*}
  where $X_a \subset \rUniverse{A = a}$ for $a \in \CA \setminus \CA'$ is a
  maximal subset of processes with $\Setcard{X_a} < \Setcard{\rUniverse{A = a}}
  / 6$.
\end{definition}
To illustrate, \cref{fig:2d_universe} shows an example of a failprone set of
each belief of a two-dimensional universe. We denote from here on by $\sqcup$
the union of sets which are specifically disjoint.

\begin{definition}[Asymmetric grid failprone system]
  \label{def:grid_asymmetric_failprone_system}
  Given $d$ attributes $A_1, \ldots A_d$, let $\CF_i$ be a failprone
  system for attribute $A_i$ as per
  \cref{def:grid_failprone_system}, for $i=1, \dots,d$.
  Then, the collection of all these
  failprone systems, $\BF = [\CF_1, \ldots, \CF_d]$, forms an \emph{asymmetric grid
  failprone system for heterogeneous processes}.
\end{definition}

\begin{definition}[Asymmetric grid quorum system for heterogeneous processes]
  \label{def:grid_asymmetric_quorum_system}
  For an asymmetric grid failprone system $\BF$, we define the \emph{asymmetric
  grid quorum system for heterogeneous processes} $\BQ$ for $\BF$ to be the
  canonical asymmetric quorum system of $\BF$, that is $\BQ = \overbar{\BF}$.
\end{definition}

To show that such a $\BQ$ is \bkconsistent{3}, we will show that the asymmetric
failprone system $\BF$ it is based on is \bkresilient{3}. To this end we first
define some auxiliary results. For an intuition on how the proof will work, see
\cref{fig:b3_proof_intuition}. Briefly, it is based on a counting argument. We
will provide upper bounds for the cardinalities of the union of full failures
of any two failprone sets as well as for the intersection of any two failprone
sets. We then consider the set of processes not yet covered by either of these,
and show that there are too many as to cover them with the partial failures of
two failprone sets.

\begin{figure}
  \centering
  \includegraphics[width=0.7\textwidth]{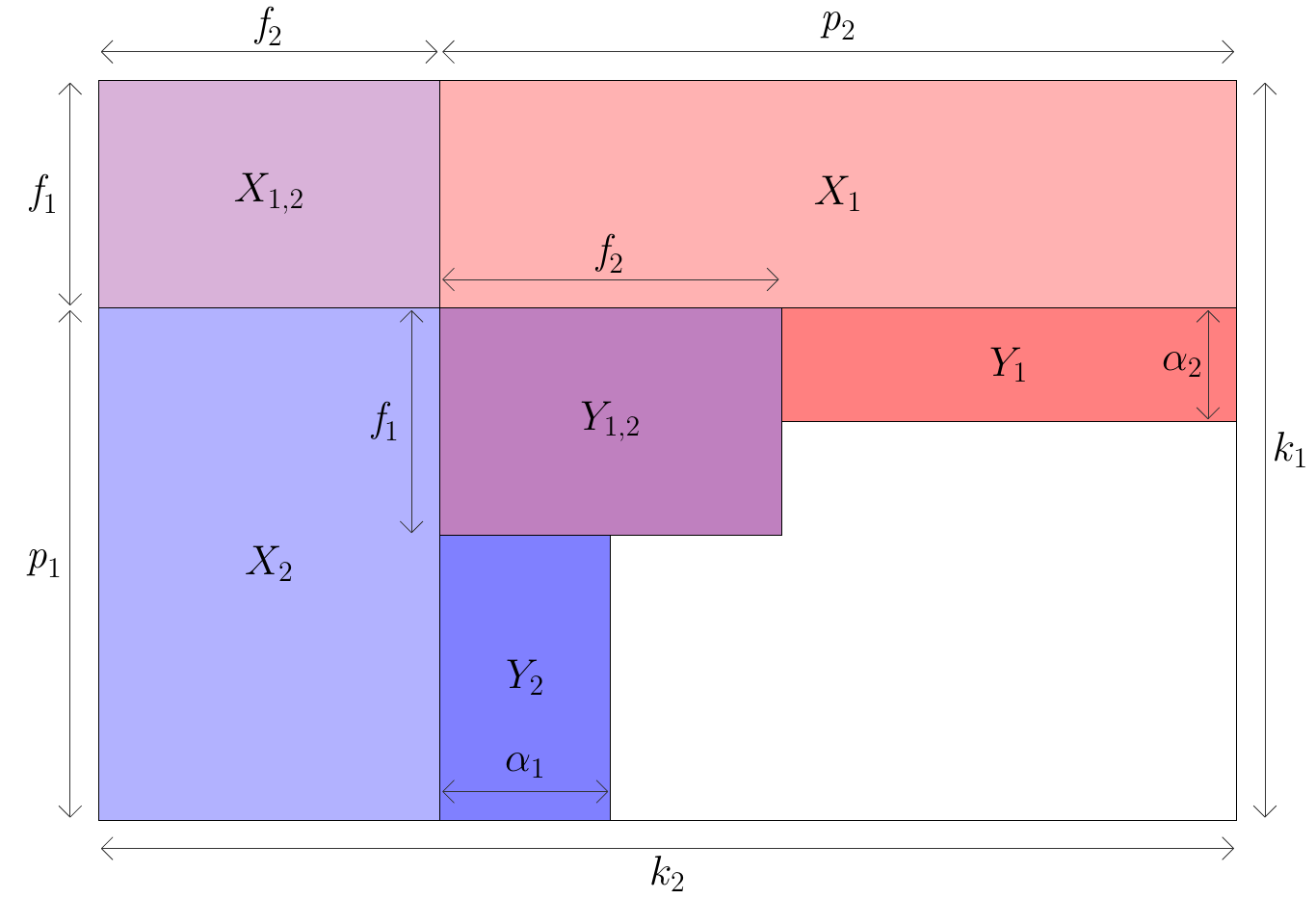}
  \caption{
    Overview of the different sets involved in proving $B^3$ consistency for a
    two-dimensional universe with attribute cardinalities $k_1$ and $k_2$. Let
    $F_1 \in \CF_1$, $F_2 \in \CF_2$ and $F_{1,2} \in \CF_1^* \cap \CF_2^*$ be
    the involved sets. 
    The components $X_1 \sqcup X_{1, 2}$ respectively $X_2 \sqcup X_{1, 2}$ are
    $\fullf(F_1)$ respectively $\fullf(F_2)$, with $X_{1, 2}$ being their
    intersection. Upper bounds for these cardinalities are provided in
    \cref{lem:cardinality_bound_full_union}.
    The components $Y_1$, $Y_2$ and $Y_{1, 2}$ make up the different parts of
    $F_{1, 2} \setminus (F_1 \cup F_2)$ when its cardinality is maximized. An
    upper bound for this cardinality is in
    \cref{lem:cardinality_bound_intersection_without_full_failures}.
    Finally, the rectangle containing components $Y_1$, $Y_2$, $Y_{1, 2}$ as
    well as the unlabeled white component, is where elements of
    $\partialf(F_1)$ and $\partialf(F_2)$ can end up. An upper bound for
    $\partialf(F_1) \cup \partialf(F_2)$ is provided in
    \cref{lem:cardinality_bound_partial_union}.
  }
  \label{fig:b3_proof_intuition}
\end{figure}

\begin{definition}
  \label{def:helper_variables}
  For any attribute $A_i$ and its corresponding failprone system $\CF_i$, we
  define
  \begin{align*}
    \varffattrs{i} = \ceil*{\frac{k_i}{3}} - 1, \hspace{1cm}
    \varpfattrs{i} = k_i - \varffattrs{i} = \floor*{\frac{2k}{3}} + 1, \text{and} \hspace{1cm} 
    \varpfprocs{i} = \ceil*{\frac{n}{6 k_i}} - 1
    .
  \end{align*}

  Hence, for some $F \in \CF_i$, the variable \varffattrs{i} is the number of
  attribute values of $A_i$ taken on by processes in $\fullf(F)$, the variable
  \varpfattrs{i} is the number of attribute values of $A_i$ taken on by
  processes in $\partialf(F)$, and lastly, \varpfprocs{i} is the number of
  processes in $\partialf(F)$ \emph{per attribute value} of $A_i$. Note that we
  can write, for some $\epsilon_{i}$ and $\delta_{i}$ in $(0, 1]$:
  \begin{align*}
    \varffattrs{i} = \frac{k_i}{3} - \epsilon_{i}, \hspace{1cm}
    \varpfattrs{i} = \frac{2 k_i}{3} + \epsilon_{i}, \hspace{1cm}
    \varpfprocs{i} = \frac{n}{6 k_i} - \delta_{i}
  \end{align*}
\end{definition}

\begin{lemma}
  \label{lem:more_full_than_partial_failures}
  Let $A_i$ and $A_j$ be any two different attributes. Then
  \[
    \frac{n}{k_i k_j} \cdot f_i > \alpha_j
    .
  \]
\end{lemma}

\begin{proof}
  %!TEX proof = {lem:more_full_than_partial_failures}

  First consider the case that $k_i \in \{4, 5\}$. Then $\varffattrs{i} = 1$, and so
  \[
    \frac{n}{k_i k_j} \cdot \varffattrs{i} 
    \geq \frac{n}{4 k_j} 
    > \frac{n}{6 k_j} 
    > \varpfprocs{j}
    .
  \]
  Otherwise, if $k_i \geq 6$, then
  \[
    \frac{n}{k_i k_j} \cdot \varffattrs{i}
    \geq \frac{n}{k_i k_j} \cdot \left(\frac{k_i}{3} - 1\right)
    \geq \frac{n}{3 k_j} - \frac{n}{6 k_j}
    > \varpfprocs{j}
    .
  \]
\end{proof}

\begin{lemma}[Cardinality of failprone set]
  \label{lem:cardinality_of_failprone_set}
  The cardinality of any failprone set $F_i \in \CF_i$ is equal to the sum of
  the number of its full- and partial-fault processes, i.e.,
  \[
    |F_i| = \frac{n}{k_i} \cdot \varffattrs{i} + \varpfattrs{i} \cdot \varpfprocs{i}
    .
  \]
\end{lemma}

\begin{proof}
  %!TEX proof = {lem:cardinality_of_failprone_set}
  By definition, $F_i = \fullf(F_i) \sqcup \partialf(F_i)$, so $|F_i| =
  |\fullf(F_i)| + |\partialf(F_i)|$, with $\fullf(F_i) = \rUniverse{A_i \in
  \CA'_i}$ for some subset of attribute values $\CA'_i \subset \CA_i$ with
  $|\CA'_i| = \varffattrs{i}$. Hence by
  \cref{lem:cardinality_of_restricted_universe_set}, $|\fullf(F_i)| =
  \frac{n}{k_i} \cdot \varffattrs{i}$. By definition, $|\partialf(F_i)| =
  \varpfattrs{i} \cdot \varpfprocs{i}$, so the claim follows.
\end{proof}

\begin{lemma}[Cardinality of union of full failures]
  \label{lem:cardinality_bound_full_union}
  Consider any two failprone sets $F_i \in \CF_i$ and $F_j \in \CF_j$ for $i
  \neq j$. Then their respective full failures necessarily overlap, and so
  \[
    \left|\fullf(F_i) \cup \fullf(F_j)\right|
    = \frac{n}{k_i} \cdot \varffattrs{i}
    + \frac{n}{k_j} \cdot \varffattrs{j} 
    - \frac{n}{k_i k_j} \cdot \varffattrs{i} \cdot \varffattrs{j}.
  \]
\end{lemma}

\begin{proof}
  %!TEX proof = {lem:cardinality_bound_full_union}
  Let $\CA'_i \subset \CA_i$ and $\CA'_j \subset \CA_j$ be the attribute values taken
  on by processes in $\fullf(F_i)$ respectively $\fullf(F_j)$. Then, by the
  definition of $\fullf(\cdot)$, and
  \cref{lem:cardinality_of_restricted_universe_set}
  \begin{align*}
    \left|\fullf(F_i) \cup \fullf(F_j)\right|
    & = \left|\fullf(F_i)\right| + \left|\fullf(F_j)\right| - \left|\fullf(F_i) \cap \fullf(F_j)\right| \\
    & = \left|\rUniverse{\CA_i \in \CA'_i}\right| 
      + \Setcard{\rUniverse{\CA_j \in \CA'_j}} 
      - \Setcard{\rUniverse{\CA_i \in \CA'_i \land \CA_j \in \CA'_j}} \\
    & = \frac{n}{k_i} \cdot \varffattrs{i} + \frac{n}{k_j} \varffattrs{j} - \frac{n}{k_i k_j} \cdot \varffattrs{i} \cdot \varffattrs{j}
    .
  \end{align*}
\end{proof}

\begin{lemma}[Union of partial failures]
  \label{lem:cardinality_bound_partial_union}
  Let $F_i \in \CF_i$ and $F_j \in \CF_j$ be any two failprone sets of
  different processes. Then
  \[
    \left|\partialf(F_i) \cup \partialf(F_j)\right| \leq \varpfattrs{i} \varpfprocs{i} + \varpfattrs{j} \varpfprocs{j}
    .
  \]
  Equality is reached by choosing the partial failures of the two such that
  they do not overlap.
\end{lemma}

\begin{proof}
  %!TEX proof = {lem:cardinality_bound_partial_union}
  By definition, the partial failures of any failprone set $F_i \in \CF_i$
  contain up to $\varpfprocs{i}$ processes for each of $\varpfattrs{i}$
  attribute values. Thus
  \[
    \left|\partialf(F_i) \cup \partialf(F_j)\right|
    \leq \left|\partialf(F_i)\right| + \left|\partialf(F_j)\right|
    = \varpfattrs{i} \varpfprocs{i} + \varpfattrs{j} \varpfprocs{j}
    .
  \]
  It is generally easy to pick the two sets disjoint, due to the low number of
  partial failures.
\end{proof}

\begin{lemma}[Intersection of two failprone sets]
  \label{lem:cardinality_bound_intersection}
  Let $F_i \in \CF_i$ and $F_j \in \CF_j$ be any two failprone sets of
  different processes. Then
  \begin{align*}
    \left|\fullf(F_i) \cap \fullf(F_j)\right| & = \frac{n}{k_i k_j} \cdot \varffattrs{i} \varffattrs{j}, \text{ and } \\
    \left|\partialf(F_i) \cap \fullf(F_j)\right| & \leq \varpfattrs{i} \varpfprocs{i}, \text{ and }\\
    \left|F_i \cap F_j\right| & \leq 
      \frac{n}{k_i k_j} \cdot \varffattrs{i} \varffattrs{j} 
      + \varpfattrs{j} \varpfprocs{j}
      + \varpfattrs{i} \varpfprocs{i}
      .
  \end{align*}
\end{lemma}

\begin{proof}
  %!TEX proof = {lem:cardinality_bound_intersection}
  Recall that by definition the full and partial failures of a failprone
  set are a partition of the set. We first consider the intersection of the
  full failures of $F_i$ and $F_j$. Let $\CA_i' \subset \CA_i$ and
  $\CA_j' \subset \CA_j$ be the sets of attribute values of $\fullf(F_i)$
  respectively $\fullf(F_j)$. Then the first expression follows by
  \cref{lem:cardinality_of_restricted_universe_set} as
  \[
    \Setcard{\fullf(F_i) \cap \fullf(F_j)} 
    = \Setcard{\rUniverse{A_i \in \CA_i' \land A_j \in \CA_j'}}
    = \frac{n}{k_i k_j} \cdot f_i f_j
    .
  \]

  The bound for $\left|\partialf(F_i) \cap \fullf(F_j)\right|$ is trivial, as
  the intersection of two sets is at most as large as either of its
  constituents, and $\left|\partialf(F_i)\right| = \varpfattrs{i}
  \varpfprocs{i}$. However, we want to show that the bound is tight. Assume
  $\fullf(F_j)$ is given, and we place the elements of $\partialf(F_i)$. For
  each attribute value $a_i \in \CA_i \setminus \CA_i'$, we can choose up to
  $\varpfprocs{i}$ processes from $\rUniverse{A_i = a_i}$ to include in the
  partial failures. We first consider the number of processes which are both in
  $\rUniverse{A_i = a_i}$ as well as in $\fullf(F_j)$, which
  by~\cref{lem:cardinality_of_restricted_universe_set} is
  \[
    \abs[\big]{\rUniverse{A_i = a_i} \cap \fullf(F_j)} 
    = \abs[big]{\rUniverse{A_i = a_i \land A_j \in \CA_j'}}
    = \frac{n}{k_i k_j} \cdot f_j
    .
  \]
  By~\cref{lem:more_full_than_partial_failures} this is larger than $\alpha_j$,
  and so the partial failures of $F_j$, rather than the full failures of $F_i$,
  will limit the size of the intersection, and so all partial failures of $F_i$
  are ``spent'' on full failures of $F_j$ and vice-versa. Then, the
  intersection $\partialf(F_i) \cap \partialf(F_j)$ is empty, and thus
  \begin{align*}
    \left|F_i \cap F_j\right|
    & \leq \left|\fullf(F_i) \cup \fullf(F_j)\right| 
         + \left|\partialf(F_i) \cup \fullf(F_j)\right|
         + \left|\fullf(F_i) \cup \partialf(F_j)\right| \\
    & \leq \frac{n}{k_i k_j} \cdot \varffattrs{i} \varffattrs{j} 
    + \varpfattrs{i} \varpfprocs{i} + \varpfattrs{j} \varpfprocs{j}
    .
  \end{align*}
  While one could place partial failures of $F_i$ and $F_j$ such that
  $\partialf(F_i) \cap \partialf(F_j) \neq \varnothing$, this would lower the
  cardinality of the intersection of partial- and full-failures by two for
  every element added to the intersection of partial failures, and so the bound
  for $\left|F_i \cap F_j\right|$ holds.
\end{proof}

\begin{lemma}
  \label{lem:cardinality_bound_intersection_without_full_failures}
  Consider any two different attributes $A_i$ and $A_j$, and let $F_i \in
  \CF_i$ and $F_j \in \CF_j$ be one failprone set of each of the respective
  failprone systems. Let $F'_{i, j} \in \CF_i^* \cap \CF_j^*$ be a
  maximum-cardinality intersection set as in
  \cref{lem:cardinality_bound_intersection}, with $F'_{i, j} = F'_i \cap F'_j$
  for some $F'_i \in \CF_i$ and $F'_j \in \CF_j$. Note that the sets $F_i$,
  $F_j$ and $F'_{i, j}$ are the three sets which show up in the definition of
  \bkconsistency{3}. Then, we provide an upper bound for the cardinality of the
  set $F'_{i, j}$ excluding those parts of $\CP$ already covered by either of
  $F_i$ and $F_j$, as
  \[
    |F'_{i, j} \setminus (F_i \cup F_j)| \leq 
      \frac{n}{k_i k_j} \cdot \varffattrs{i} \varffattrs{j}
      + (\varpfattrs{j} - \varffattrs{j}) \varpfprocs{j}
      + (\varpfattrs{i} - \varffattrs{i}) \varpfprocs{i}
      .
  \]
\end{lemma}

\begin{proof}
  %!TEX proof = {lem:cardinality_bound_intersection_without_full_failures}
  Recall that for such a maximum-cardinality intersection set, the constituent
  sets $\partialf(F'_i)$ and $\partialf(F'_j)$ are disjoint, so we can
  partition $F'_{i, j}$ as
  \begin{align*}
    F'_{i ,j} \setminus (F_i \cup F_j) = 
         & \left(\fullf(F'_i) \cap \fullf(F'_j) \setminus (F_i \cup F_j)\right) \\
    & \sqcup \left(\partialf(F'_i) \cap \fullf(F'_j) \setminus (F_i \cup F_j)\right) \\
    & \sqcup \left(\fullf(F_i) \cap \partialf(F'_j) \setminus (F_i \cup F_j)\right)
    .
  \end{align*}

  To maximize the cardinality of this set, we want to minimize the intersection of $F'_{i,
  j}$ with the two failprone sets $F_i$ and $F_j$. We first consider the set
  $\left(\fullf(F'_i) \cap \fullf(F'_j)\right) \setminus (F_i \cup F_j)$. A
  trivial bound for its cardinality is
  $\left|\left(\fullf(F'_i) \cap \fullf(F'_j)\right) 
  \setminus (F_i \cup F_j)\right| 
  \leq \left|\fullf(F'_i) \cap \fullf(F'_j)\right|$,
  and indeed we cannot do any better. To give an intuition for this, we can
  pick the underlying attribute values of the full failures of $F_i$ and $F_j$
  disjoint from those of $F'_i$ and $F'_j$. Then, the sets $\fullf(F_i) \cup
  \fullf(F_j)$ and $\fullf(F'_i) \cap \fullf(F'_j)$ are disjoint. It is further
  easy to pick the partial failures of $F_i$ and $F_j$ such that they do not
  overlap either, and thus equality holds.
  Hence, by \cref{lem:cardinality_bound_intersection}
  \[
    \left|\left(\fullf(F'_i) \cap \fullf(F'_j)\right) \setminus \left(F_i \cup F_j\right)\right|
    \leq \left|\left(\fullf(F'_i) \cap \fullf(F'_j)\right)\right|
    = \frac{n}{k_i k_j} \cdot \varffattrs{i} \varffattrs{j}
    .
  \]

  For the intersection $\partialf(F'_i) \cap \fullf(F'_j)$ we cannot
  avoid overlapping with $\fullf(F_i)$, as some of the partial failures of $F'_i$
  have to take on attribute values also taken on by processes in $\fullf(F_i)$.
  By definition, this is the case for $f_i$ attribute values of attribute
  $A_i$. Adapting our bound from \cref{lem:cardinality_bound_intersection} we
  get
  \[
    \left|\partialf(F'_i) \cap \fullf(F'_j) \setminus (F_i \cup F_j)\right|
    \leq \varpfattrs{i} \varpfprocs{i} - \varffattrs{i} \varpfprocs{i} 
    = (\varpfattrs{i} - \varffattrs{i}) \varpfprocs{i}
    .
  \]
  And by symmetry,
  $\left|\partialf(F'_j) \cap \fullf(F'_i) \setminus (F_i \cup F_j)\right|
    \leq (\varpfattrs{j} - \varffattrs{j}) \varpfprocs{j}$.
  Then, the claim follows as the sum of these three bounds.
\end{proof}

\begin{theorem}[\qkresilience{3} of $\CF_i$]
  \label{th:q3_resilience}
  Every failprone system $\CF_i \in \BF$ is \qkresilient{3}.
\end{theorem}

\begin{proof}
  %!TEX proof = {th:q3_resilience}
  Let $F_1, F_2, F_3 \in \CF_i$ be any three failprone sets of $\CF_i$.
  Consider their union $F_1 \cup F_2 \cup F_3$. The full failures of each will
  cover strictly less than $1/3$ of the values of attribute $A_i$, so there
  exists at least one attribute value $a \in A_i$ outside the union
  $\fullf(F_1) \cup \fullf(F_2) \cup \fullf(F_3)$. As each failprone set's
  partial failures contains strictly less than $1/6$ of processes per attribute
  value, the union $\partialf(F_1) \cup \partialf(F_2) \cup \partialf(F_3)$
  will cover less than $1/2$ of $\rUniverse{A_i = a}$, and so $F_1 \cup F_2
  \cup F_3 \subsetneq \CP$.
\end{proof}

\begin{theorem}[\bkresilience{3} of $\BF$]
  \label{th:b3_resilience}
  The asymmetric failprone system $\BF = [\CF_1, \ldots, \CF_d]$ is \bkresilient{3}.
\end{theorem}

\begin{proof}
  %!TEX proof = {th:b3_resilience}
  By \cref{th:q3_resilience}, each $\CF_i \in \BF$ is \qkresilient{3}.
  Consider thus any two $\CF_i, \CF_j \in \BF$ with $i \neq j$. Let $F_i \in
  \CF_i$ and $F_j \in \CF_j$ be one failprone set of each, and $F_{i, j} \in
  \CF^*_i \cap \CF^*_j$ a joint failprone set of the two. In order to show that
  $\BF$ is \bkresilient{3}, we show that $\abs{F_i \cup F_j \cup F_{i, j}} <
  n$. To do so, we first partition this set and then bound the cardinality of each
  part using earlier results.
  \begin{align*}
    \left|F_i \cup F_j \cup F_{i, j}\right|
    =    & \left|F_i \cup F_j \sqcup (F_{i, j} \setminus (F_i \cup F_j))\right| \\
    \leq & \left|\fullf(F_i) \cup \fullf(F_j)\right|
           + \left|\partialf(F_i) \cup \partialf(F_j)\right| \\
         & + \left|F_{i, j} \setminus (F_i \cup F_j)\right|
         && \text{By definition} \\
    \leq & \frac{n}{k_i} \cdot \varffattrs{i} 
           + \frac{n}{k_j} \cdot \varffattrs{j} 
           - \frac{n}{k_i k_j} \cdot \varffattrs{i} \varffattrs{j}
         && \text{\Cref{lem:cardinality_bound_full_union}} \\
         & + \varpfattrs{i} \varpfprocs{i} + \varpfattrs{j} \varpfprocs{j} 
         && \text{\Cref{lem:cardinality_bound_partial_union}} \\
         & + \frac{n}{k_i k_j} \cdot \varffattrs{i} \varffattrs{j}
           + (\varpfattrs{i} - \varffattrs{i}) \varpfprocs{i}
           + (\varpfattrs{j} - \varffattrs{j}) \varpfprocs{j} 
         && \text{\Cref{lem:cardinality_bound_intersection_without_full_failures}}
    \end{align*}
    We now insert the expressions of \cref{def:helper_variables} and
    simplify. As all $\epsilon$ and $\delta$ values are strictly positive, it
    follows that
    \begin{align*}
      \left|F_i \cup F_j \cup F_{i, j}\right| & 
      \leq n 
         - k_i \delta_{i} - k_j \delta_{j} 
         - \frac{n}{2 k_i} \cdot \epsilon_{i}
         - \frac{n}{2 k_j} \cdot \epsilon_{j}
         - 3 (\epsilon_{i} \delta_{i} + \epsilon_{j} \delta{j})
       < n
       .
    \end{align*}
    And so $\BF$ is \bkresilient{3}.
\end{proof}

\paragraph*{Tightness of the construction}
\label{sec:tightness}

We briefly discuss tightness of the construction in the sense of whether its
parameters can be increased without violating \bkresilience{3}. The two
parameters to consider are \varffattrs{i}, the ratio of attribute values making
up the full failures of a failprone set, and \varpfprocs{i}, the number of
processes per remaining attribute value which make up the partial failures.

\begin{lemma}
  \label{lem:tightness_full_failures}
  The construction is tight with regards to the ratio of full-failure
  attribute values it supports.
\end{lemma}

\begin{proof}
  For the ratio of full-failure attribute values we use $\varffattrs{i} =
  \ceil*{k_i / 3}$. Assume we instead used some value $\varffattrsbar{i} >
  \varffattrs{i}$. As these are integers, it must hold that $\varffattrsbar{i}
  \geq \varffattrs{i} + 1 \geq k_i / 3$. Let now $F'_1, F'_2, F'_3 \in
  \CF'$ be three failprone sets of such a hypothetical failprone system. Recall
  that we can choose the attribute values of their full failures freely. But
  then as $\varffattrsbar{i} > k_i / 3$, it is easy to choose those such that
  the union $F'_1 \cup F'_2 \cup F'_3$ covers the universe $\CP$, and the
  system would not be \bkconsistent{3}.
\end{proof}

For $\varpfprocs{i}$, the second parameter, the value we chose is
asymptotically required for our proof of consistency in
\cref{th:b3_resilience}. However, it is not optimal due to two reasons. First,
it is not possible to simultaneously achieve both a maximum-cardinality
intersection set as in \cref{lem:cardinality_bound_intersection} and a
maximum-cardinality set of partial failures as in
\cref{lem:cardinality_bound_partial_union}, as a maximum-cardinality
intersection set will inevitably fully cover all processes for some attribute
values, overlapping with the set of partial failures. Empirically, to maximize
the cardinality of a union as seen in the \bkresilience{3} proof, one will end
up with a tradeoff where the intersection set is chosen slightly non-maximal to
permit making use of all partial failures. Second, for specific choices of
attribute cardinalities, the \varpfprocs{i} parameter can actually be
incremented while remaining compatible with our proof.
\Cref{fig:tightness_alpha} (\cref{app:tightness}) visualizes this for the
two-dimensional case, showing that this mainly matters for low-attribute
cardinalities.

\section{Comparison with threshold failure assumption}
\label{sec:comparison_with_threshold}

Asymmetric trust models such as the one presented here are
more expressive than traditional symmetric failure assumptions, in that they
allow processes to choose failure assumptions based on personal beliefs.
However, protocols that use them are closely related to the traditional ones,
and so expressiveness it not an end in itself. Rather, processes strive to
perform some task in spite of other processes being faulty. In order to be
useful, asymmetric failure assumptions must be able to model some failure
scenarios ---~at least for a subset of processes~--- which symmetric
assumptions could not. We will discuss this by comparing our failprone systems
with the commonly-used \emph{threshold failure assumptions}.

The threshold failure assumption of interest to us is $n > 3f$, which states
that strictly less than one third of processes are faulty. This quorum system
is \qkconsistent{3}, which is the symmetric analogue of our \bkconsistent{3}
construction. We denote by $\CF_{n > 3f}$ the failprone system of such a
threshold failure assumption, i.e.,
\[
  \CF_{n > 3f} = \left\{X \subset \CP : \left|X\right| = \ceil*{\frac{n}{3} - 1}\right\}
  .
\]

Consider an asymmetric grid failprone system $\BF$ as per
\cref{def:grid_asymmetric_failprone_system}. By
\cref{lem:cardinality_of_failprone_set}, all failprone sets of any given
failprone system $\CF_i \in \BF$ have the same cardinality. If these failprone
systems would consist of failprone sets with cardinality at most
$\ceil*{\frac{n}{3}} - 1$, then none of these systems could model a failure
scenario which is not also modeled by $\CF_{n > 3f}$, and no rational process
would be inclined to use it. We thus only care about asymmetric failprone
systems where there are at least some failprone sets with a cardinality larger
than those of the $n > 3f$ threshold failure assumption, and call such systems
\emph{useful}.

\begin{definition}[Usefulness of asymmetric grid failprone systems]
  An asymmetric grid failprone system $\BF$ is \emph{useful}, if there exists
  at least one failprone system $\CF \in \BF$, whose failprone sets have a
  cardinality larger than $\ceil*{\frac{n}{3}} - 1$.
\end{definition}
In the remainder, we characterize which universes permit for \emph{useful}
asymmetric grid failprone systems. Many of the involved arguments are
straightforward and have been relegated to the appendix.

\paragraph*{Equal-cardinality attributes}

Consider first the special case of all attributes having cardinality $k$, for a
total of $n = k^d$ processes. By \cref{lem:cardinality_of_failprone_set}, the
cardinality of any failprone set $F_i$ is
\[
  \Setcard{F_i}
  = \ceil*{\frac{k}{3} - 1} \cdot \frac{n}{k} 
  % Not using auto-sizing for brackets of last term to ensure consistent size.
  + \floor*{\frac{2k}{3} + 1} \cdot \ceil[\bigg]{\frac{n}{6k} - 1}
  .
\]

Since this cardinality approaches $\abs{F_i} = \frac{n}{3} + \frac{n}{9}$ as
$k$ grows, it is strictly larger than the asymptotic cardinality $\frac{n}{3}$
of the threshold assumption. For smaller values of $k$ this may not be the
case, however. We show that for the two-dimensional case, all attribute
cardinalities above $7$, except for $k \in \{9, 12\}$, permit for useful
failprone systems to emerge. With the straight-forward optimization of
$\varpfprocs{\cdot}$ as discussed in \cref{sec:tightness}, this even extends to
$k = 5$. Effectively this means that most universes with at least $25$
processes are useful for our construction.
For the $d$-dimensional case, the same can be shown for all cardinalities above
$4$, except for $k = 6$. This covers many possible universes from $64$
processes onward. Details, including proofs, can be found in
\cref{tbl:useful_equal_card} (\cref{app:comparison}).

\paragraph*{Unequal-cardinality attributes}

The analysis is more complex if we allow attributes to differ in their
cardinalities, since failprone sets of two failprone systems may not have the
same cardinality. For a concrete set of attributes it is possible that one
failprone system is useful while another is not. As above, we aim to
characterize which parts of the parameter space permit for useful failprone
systems to exist.

\Cref{fig:comparison_with_threshold_2d} (\cref{app:comparison}) shows
the ratio of grid-style to threshold-style failprone sets for the
two-dimensional case, and highlights some of the patterns which emerge. As an
example, there are already configurations which permit for failprone systems to
be useful with cardinalities as low as $k_1 = 4$ and $k_2 = 7$, for a total of
$28$ processes. If one wants both failprone systems to be useful, then most
configurations where both attributes' cardinalities are above $7$ are suitable,
except some values where one of the two is equal to $8$, $9$ or $12$.

For higher dimensions, useful failprone systems emerge for e.g. $d \geq 3$ and
$k_i \in \{4, 7, 8\}$, as well as for virtually all cases where attributes have
a higher cardinality. We provide more details and proofs in
\cref{tbl:useful_nequal_card} in the appendix. A more generic comparison with
threshold quorums is left for future work.

\section{Conclusion}
\label{sec:conclusion}

We have introduced a quorum system construction for the asymmetric trust
setting with byzantine faults which allows processes to express their
subjective assumptions without the need for coordination. This breaks the
cyclic dependency that is otherwise present in this setting, where
processes need to coordinate among themselves before being able to participate
in the distributed protocol. Our construction offers processes a choice among
different quorum systems based on attributes in which real-world processes
commonly differ. Each process can then pick the one which aligns most with its
subjective view.

We have compared these quorum systems with the threshold failure assumption
$n > 3f$ used commonly, and have shown that they are viable, and can provide
higher resilience, for different configurations of processes, from as little as
$25$ processes, but especially performing well when the number of processes
increases. There are clear avenues for future work, such as allowing for
``gaps'' in the universe of processes, an optimization of parameters, or
practical benchmarks.

\section*{Acknowledgments}

This work was supported by the Swiss National Science Foundation (SNSF)
under grant agreement Nr\@.~219403 (Emerging Consensus).

\printbibliography

\newpage

\appendix

\section{Tightness of the construction}
\label{app:tightness}

\begin{figure}[ht]
  \centering
  \includegraphics[width=0.7\linewidth]{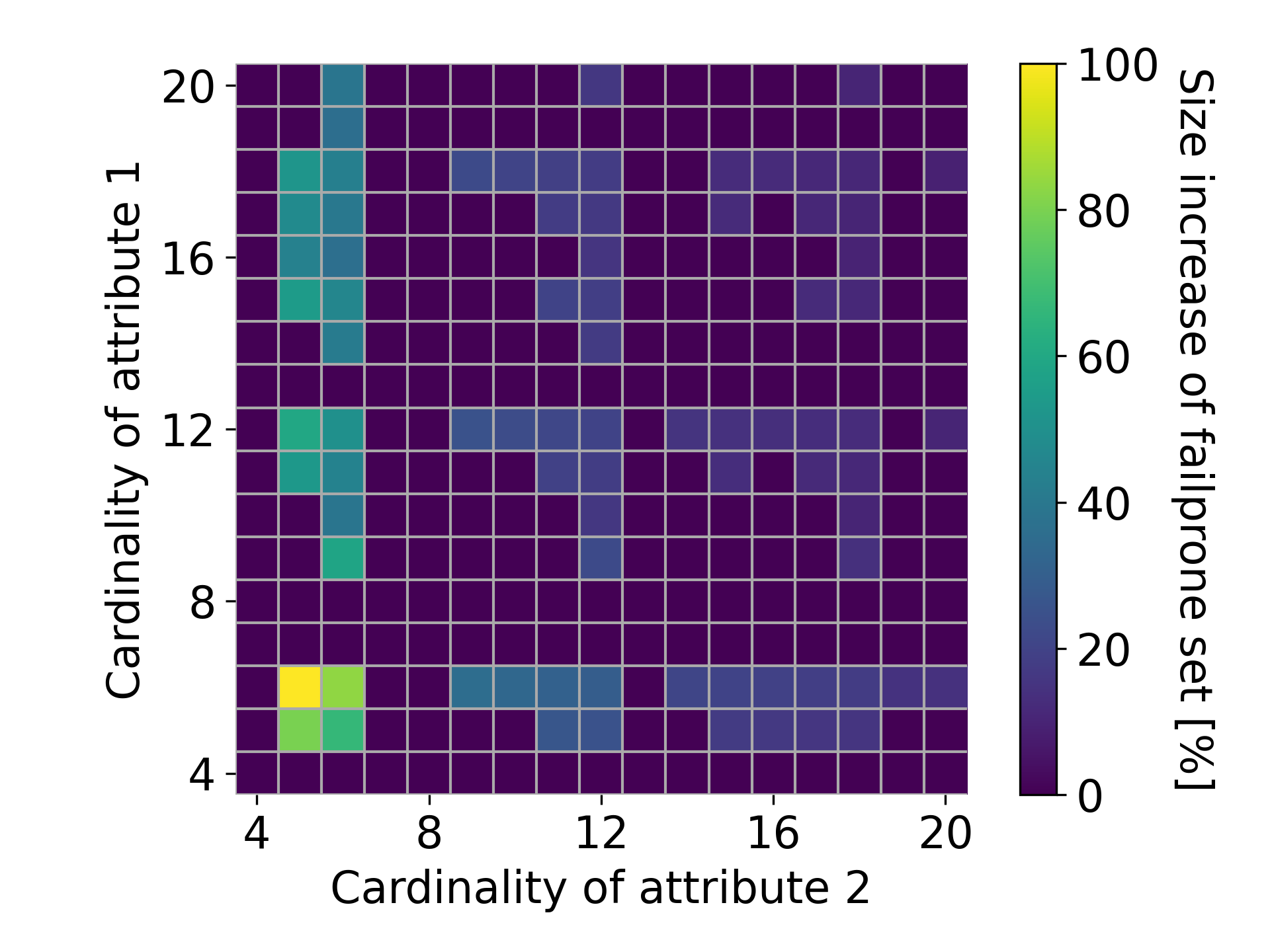}
  \caption{
    Percentage-wise size increase of grid failprone sets of attribute $A_1$
    when using the numerically-determined highest possible value for
    $\varpfprocs{i}$, compared to the one it uses by definition. This shows
    that the lack of tightness in the definition of $\varpfprocs{i}$ is mostly
    impactful if cardinality of $A_2$ is low, whereas it quickly becomes
    negligible as cardinality of the attributes increases.
  }
  \label{fig:tightness_alpha}
\end{figure}

\section{Comparison with Threshold Quorums}
\label{app:comparison}

\paragraph*{Equal-cardinality attributes}

\begin{table}[ht]
  \centering
  \caption{Parameter space of universes permitting for useful asymmetric grid
  failprone system when all attributes have equal cardinality.}
  \begin{tabular}{lll}
    \toprule
    Dimensionality $d$ & Attribute cardinality $k$ & Proof \\
    \midrule
    $d = 2$             & $k = 5$                          & Numerical, with $\varpfprocs{\cdot} = 1$ \\
    $d = 2$             & $k \in \{7, 8, 10, 11, 13, 14\}$ & Numerical \\
    $d = 2$             & $k \geq 15$                      & \cref{lem:card_bound_two_dim_k_15} \\
    $d \in \{3, 4, 5\}$ & $k = 5$                          & Numerical, with optimized $\varpfprocs{\cdot}$ \\
    $d \geq 3$          & $k = 4$                          & \cref{lem:card_bound_d_dim_k_4} \\
    $d \geq 3$          & $k \in \{7, 8, 9\}$              & \cref{lem:card_bound_d_dim_k_789} \\
    $d \geq 3$          & $k \geq 10$                      & \cref{lem:card_bound_d_dim_k_10} \\
    \bottomrule
  \end{tabular}
  \label{tbl:useful_equal_card}
\end{table}

Recall that by \cref{lem:cardinality_of_failprone_set}, the cardinality of
every failprone set if all $d$ attributes have cardinality $k$ is
\[
  \Setcard{F_i}
  = \ceil*{\frac{k}{3} - 1} \cdot \frac{n}{k} 
  % Not using auto-sizing for brackets of last term to ensure consistent size.
  + \floor*{\frac{2k}{3} + 1} \cdot \ceil[\bigg]{\frac{n}{6k} - 1}
  .
\]

\begin{lemma}
  Let $k \geq 15$ be a positive integer. Let $n = k^2$. Then
  \[
    \ceil*{\frac{k}{3} - 1} \cdot k + \floor*{\frac{2k}{3} + 1} \cdot \ceil*{\frac{k}{6} - 1} 
  > \ceil[\bigg]{\frac{n}{3} - 1}
    .
  \]
  \label{lem:card_bound_two_dim_k_15}
\end{lemma}

\begin{proof}
  \begin{align*}
    \ceil*{\frac{k}{3} - 1} \cdot k + \floor*{\frac{2k}{3} + 1} \cdot \ceil*{\frac{k}{6} - 1}
    & \geq \left(\frac{k}{3} - 1\right) \cdot k + \frac{2k}{3} \cdot \left(\frac{k}{6} - 1\right)
    = \frac{4k^2}{9} - \frac{5k}{3} \\
    & \geq \frac{k^2}{3} && k \geq 15 \\
    & > \ceil[\bigg]{\frac{n}{3} - 1}
  \end{align*}
\end{proof}

\begin{lemma}
  Let $d \geq 3$ be a positive integer. Let $n = 4^d$. Then
  \[
    \ceil*{\frac{4}{3} - 1} \cdot 4^{d-1} + \floor*{\frac{8}{3} + 1} \cdot \ceil*{\frac{4^{d-1}}{6} - 1} 
  > \ceil[\bigg]{\frac{n}{3} - 1}
    .
  \]
  \label{lem:card_bound_d_dim_k_4}
\end{lemma}

\begin{proof}
  Consider first $d = 3$. Then, it holds that $\varffattrs{i} = 1$,
  $\varpfattrs{i} = 4$ and $\varpfprocs{i} = 2$, so our failprone sets have a
  cardinality of $1 \cdot 16 + 3 \cdot 2 = 22 > 21$.
  Now for arbitrary $d > 3$,
  \begin{align*}
    \ceil*{\frac{4}{3} - 1} \cdot 4^{d-1} + \floor*{\frac{8}{3} + 1} \cdot \ceil*{\frac{4^{d-1}}{6} - 1} 
    & = 4^{d-1} + 3 \cdot \ceil*{\frac{4^{d-1}}{6}} - 3 \\ 
    & \geq 4^{d-1} + 3 \cdot \frac{4^{d-1}}{6} - 3 \\ 
    & = \frac{3}{2} \cdot 4^{d-1} - 3 \\
    & = \frac{4^d}{3} + \frac{4^d}{24} - 3 \\
    & \geq \frac{4^d}{3} && d > 3 \\
    & > \ceil[\bigg]{\frac{n}{3} - 1}
    .
  \end{align*}
\end{proof}

\begin{lemma}
  Let $d \geq 3$ be a positive integer. Let $n = 5^d$. Then
  \[
    \ceil*{\frac{5}{3} - 1} \cdot 5^{d-1} + \floor*{\frac{10}{3} + 1} \cdot \ceil*{\frac{5^{d-1}}{6} - 1} 
  > \ceil[\bigg]{\frac{n}{3} - 1}
    .
  \]
  \label{lem:card_bound_d_dim_k_5}
\end{lemma}

\begin{lemma}
  Let $d \geq 3$ be a positive integer and $k \in \{7, 8, 9\}$. Let $n = k^d$. Then
  \[
    \ceil*{\frac{k}{3} - 1} \cdot k^{d-1} + \floor*{\frac{2k}{3} + 1} \cdot \ceil*{\frac{k^{d-1}}{6} - 1} 
  > \ceil[\bigg]{\frac{n}{3} - 1}
    .
  \]
  \label{lem:card_bound_d_dim_k_789}
\end{lemma}

\begin{proof} We follow the structure of the proof for
  \cref{lem:card_bound_d_dim_k_4}. Note first that for $k \in \{7, 8, 9\}$, it
  holds that $\ceil*{k/3 -1} = 2$, and $\floor*{2k/3 + 1} = k-2$. Then
  \begin{align*}
    \ceil*{\frac{k}{3} - 1} \cdot k^{d-1} & + \floor*{\frac{2k}{3} + 1} \cdot \ceil*{\frac{k^{d-1}}{6} - 1}  \\
    & \geq 2 \cdot k^{d-1} + (k-2) \cdot \frac{4^{d-1}}{6} - (k-2) \\ 
    & = \frac{(12 + k - 2) \cdot k^{d-1}}{6} - (k-2) \\
    & = \frac{(10 + k) \cdot k^{d}}{6k} - (k-2) \\
    & = \frac{k^d}{3} + \frac{(10 - k)}{6k} \cdot k^d - (k-2) \\
    & > \frac{k^d}{3} \geq \ceil[\bigg]{\frac{n}{3} - 1} && k \in \{7, 8, 9\} \text{ and } d > 2
  \end{align*}
\end{proof}

\begin{lemma}
  Let $k \geq 10$ and $d \geq 3$ be positive integers. Let $n = k^d$. Then
  \[
    \ceil*{\frac{k}{3} - 1} \cdot k^{d-1} + \floor*{\frac{2k}{3} + 1} \cdot \ceil*{\frac{k^{d-1}}{6} - 1} 
  > \ceil[\bigg]{\frac{n}{3} - 1}
    .
  \]
  \label{lem:card_bound_d_dim_k_10}
\end{lemma}

\begin{proof}
  \begin{align*}
    \ceil*{\frac{k}{3} - 1} \cdot k^{d-1} & + \floor*{\frac{2k}{3} + 1} \cdot \ceil*{\frac{k^{d-1}}{6} - 1} \\
    & \geq \left(\frac{k}{3} - 1\right) \cdot k^{d-1} + \frac{2k}{3} \cdot \left(\frac{k^{d-1}}{6} - 1\right) \\
    % & = \frac{4k^d}{9} - k^{d-1} - \frac{2k}{3} \\
      & = \frac{k^d}{3} + \frac{k}{9} \cdot k^{d-1} - k^{d-1} - \frac{2k}{3} \\
      & \geq \frac{k^d}{3} + \frac{1}{9} \cdot k^2 - \frac{2k}{3} 
      && k \geq 10 \text{ and } d > 2 \\
      & = \frac{k^d}{3} + k \cdot \left(\frac{k - 6}{9}\right) \\
      & > \frac{k^d}{3} > \ceil[\bigg]{\frac{n}{3} - 1} && k \geq 10
  \end{align*}
\end{proof}

\paragraph*{Unequal-cardinality attributes}

\begin{figure}
  \centering
  \includegraphics[width=0.7\linewidth]{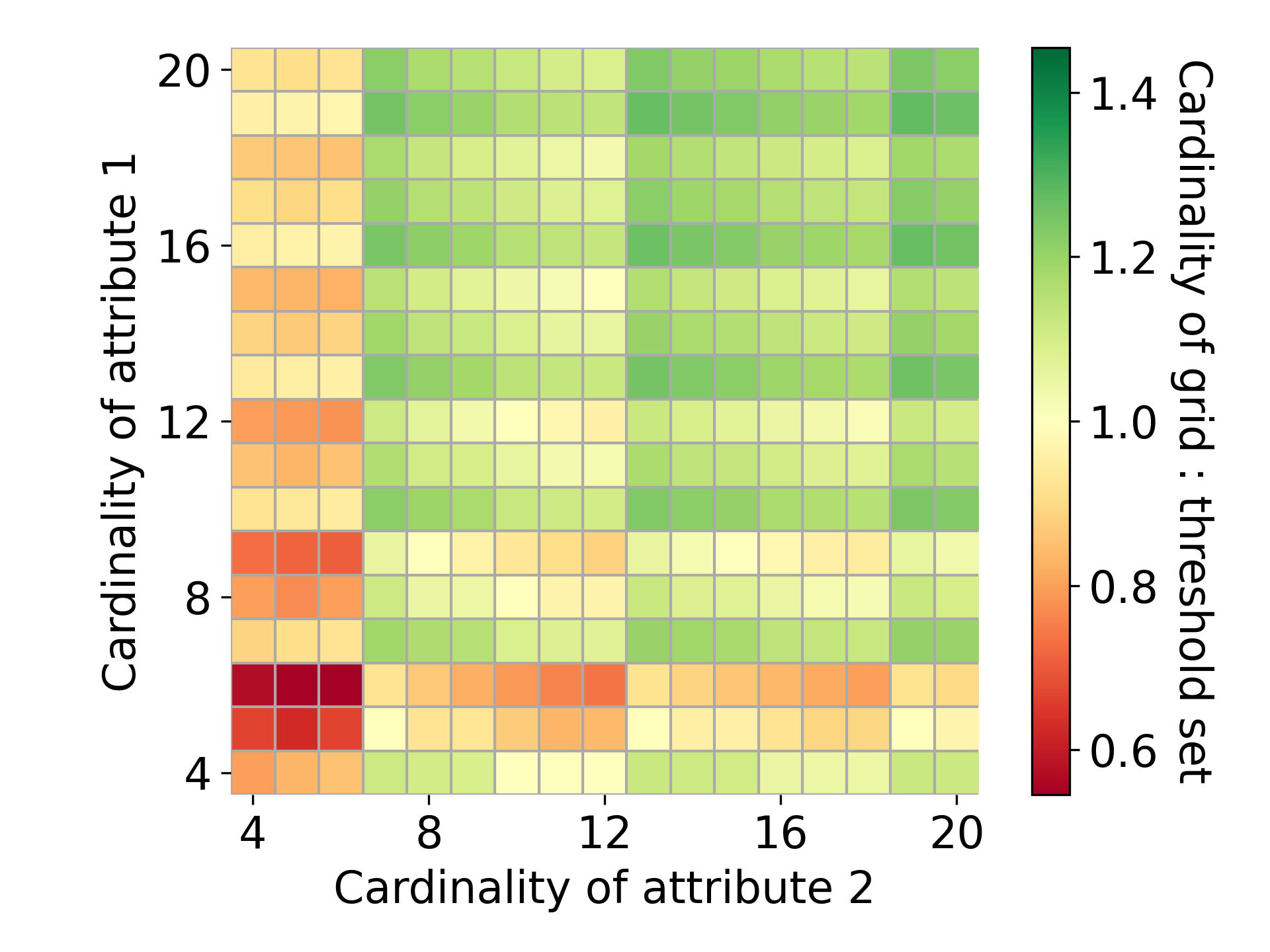}
  \caption{Ratio of the cardinality of grid failprone sets to threshold
    failprone sets for a two-dimensional universe. Each square is one possible
    two-dimensional universe $\CP = \CA_1 \times \CA_2$. The color-coded ratio
    is the size of the failprone set of the row belief divided by the size
    of a threshold failprone set for the same universe.
  }
  \label{fig:comparison_with_threshold_2d}
\end{figure}

\begin{table}[ht]
  \centering
  \caption{%
    Excerpt of parameter space of universes permitting for useful grid
    failprone system $\CF_1$ when attributes have different cardinalities.
}
\begin{tabular}{lll}
  \toprule
  Dimensionality $d$ & Attribute cardinalities $k_1, \ldots, k_d$ & Proof \\
  \midrule
  $d = 2$ & $k_1 = 4$, $k_2 \in \{7, 8, 9\}$ & Numerical \\
  $d = 2$ & $k_1 = 4$, $k_2 \geq 13$ & \cref{lem:card_neq_two_dim_k_4_13}\\
  $d = 2$ & $k_1 = 7$, $k_2 \geq 7$ & \cref{lem:card_neq_two_dim_k_7_7} \\
  $d \geq 3$ & $k_1 = 4$ & \cref{lem:card_neq_d_dim_k_4} \\
  $d \geq 3$ & $k_1 = 7$ & \cref{lem:card_neq_d_dim_k_7} \\
  $d \geq 3$ & $k_1 = 8$ & \cref{lem:card_neq_d_dim_k_8} \\
  \bottomrule
\end{tabular}
\label{tbl:useful_nequal_card}
\end{table}

Recall that by \cref{lem:cardinality_of_failprone_set}, the cardinality of
every failprone set $F$ of failprone system $\CF_i$, if the $d$ attributes have
cardinalities $k_1, \ldots, k_d$, is
\[
  \Setcard{F_i}
  = \ceil*{\frac{k_i}{3} - 1} \cdot \frac{n}{k_i} 
  + \floor*{\frac{2k_i}{3} + 1} \cdot \ceil*{\frac{n}{6k_i} - 1}.
\]

\begin{lemma}
  Let $k_1 = 4$ and $k_2 \geq 13$ be positive integers and $n = k_1 k_2$. Then
  \[
    k_2 + 3 \ceil*{\frac{k_2}{6} - 1} > \ceil[\bigg]{\frac{n}{3} - 1}
    .
  \]
  \label{lem:card_neq_two_dim_k_4_13}
\end{lemma}

\begin{proof}
  Consider first $18 > k_2 \geq 13$. Then, $\ceil*{\frac{k_2}{6} - 1} = 2$, and so
  \[
    k_2 + 3 \ceil*{\frac{k_2}{6} - 1} = k_2 + 6 > \frac{4k_2}{3} = \frac{n}{3} > \ceil[\bigg]{\frac{n}{3} - 1}
    .
  \]
  Now let $k_2 \geq 18$. Then
  \begin{align*}
    k_2 + 3 \ceil*{\frac{k_2}{6} - 1} 
      > k_2 + \frac{k_2}{2} - 3 
      % = \frac{4 k_2}{3} + \frac{k_2}{6} - 3
      > \frac{4 k_2}{3} 
      > \ceil[\bigg]{\frac{n}{3} - 1}
      .
  \end{align*}
\end{proof}

\begin{lemma}
  Let $k_1 = 7$ and $k_2 \geq 7$ be positive integers. Let $n = k_1 k_2$. Then
  \[
    2 k_2 + 5 \ceil*{\frac{k_2}{6} - 1} > \ceil*{\frac{n}{3} - 1}
    .
  \]
  \label{lem:card_neq_two_dim_k_7_7}
\end{lemma}

\begin{proof}
  Consider first $10 > k_2 \geq 7$. Then, $\ceil*{\frac{k_2}{6} - 1} = 1$, and so
  \[
    2 k_2 + 5 \ceil*{\frac{k_2}{6} - 1} 
    = 2 k_2 + 5
    \geq 19
    > \frac{7 k_2}{3}
    > \ceil*{\frac{n}{3} - 1}
    .
  \]
  Now let $k_2 \geq 10$. Then
  \[
    2 k_2 + 5 \ceil*{\frac{k_2}{6} - 1} 
    > 2 k_2 + \frac{5 k_2}{6} - 5 
    %= \frac{7 k_2}{3} + \frac{3 k_2}{6} - 5
    > \frac{7 k_2}{3}
    > \ceil{\frac{n}{3} - 1}
    .
  \]
\end{proof}

\begin{lemma}
  Let $k_1 = 4$ and $k_2, \ldots, k_d \geq 4$ for $d \geq 3$ be positive
  integers. Let $n = \prod_{i = 1}^d k_i$. Then
  \[
    \frac{n}{4} + 3 \cdot \ceil*{\frac{n}{24} - 1} > \ceil*{\frac{n}{3} - 1}
    .
  \]
  \label{lem:card_neq_d_dim_k_4}
\end{lemma}

\begin{proof}
  Consider first $d = 3$, and $k_1 = k_2 = k_3 = 4$. Then the result follows
  numerically:
  \[
    \frac{n}{4} + 3 \cdot \ceil*{\frac{n}{24} - 1} 
    = 16 + 3 \cdot 2
    = 22
    > 21
    = \ceil*{\frac{n}{3} - 1}
    .
  \]
  Assume from now that either $d \geq 4$, or at least one of $k_2 \geq 5$ or $k_3 \geq 5$. Now
  \begin{align*}
    \frac{n}{4} + 3 \cdot \ceil*{\frac{n}{24} - 1} 
    & \geq \frac{n}{4} + \frac{n}{8} - 3
      % = \frac{n}{3} + \frac{n}{24} - 3
      = \frac{n}{3} + \frac{n}{24} - 3
      .
  \end{align*}
  If $d \geq 4$, then we are done, as
  \[
      \frac{n}{3} + \frac{n}{24} - 3
      \geq \frac{n}{3} + \frac{4^4}{24} - 3
      > \frac{n}{3}
      > \ceil*{\frac{n}{3} - 1}
      .
  \]
  If instead $d = 3$ and either $k_2 \geq 5$ or $k_3 \geq 5$ we are done too, as
  \[
      \frac{n}{3} + \frac{n}{24} - 3
      \geq \frac{n}{3} + \frac{4^2 \cdot 5}{24} - 3
      > \frac{n}{3}
      > \ceil*{\frac{n}{3} - 1}
      .
  \]
\end{proof}

\begin{lemma}
  Let $k_1 = 7$ and $k_2, \ldots, k_d \geq 4$ for $d \geq 3$ be positive
  integers. Let $n = \prod_{i = 1}^d k_i$. Then
  \[
    \frac{2n}{7} + 5 \cdot \ceil*{\frac{n}{42} - 1} > \ceil*{\frac{n}{3} - 1}
    .
  \]
  \label{lem:card_neq_d_dim_k_7}
\end{lemma}

\begin{proof}
  \begin{align*}
    \frac{2n}{7} + 5 \cdot \ceil*{\frac{n}{42} - 1} 
    > \frac{2n}{7} + \frac{5n}{42} - 5
    = \frac{n}{3} + \frac{n}{14} - 5
    \geq \frac{n}{3} + \frac{8 \cdot 4^2}{14} - 5
    > \frac{n}{3}
    > \ceil*{\frac{n}{3} - 1}
    .
  \end{align*}
\end{proof}

\begin{lemma}
  Let $k_1 = 8$ and $k_2, \ldots, k_d \geq 4$ for $d \geq 3$ be positive
  integers. Let $n = \prod_{i = 1}^d k_i$. Then
  \[
    \frac{2n}{8} + 6 \cdot \ceil*{\frac{n}{48} - 1} > \ceil*{\frac{n}{3} - 1}
    .
  \]
  \label{lem:card_neq_d_dim_k_8}
\end{lemma}

\begin{proof}
  Consider first $d = 3$, $k_1 = 8$ and $k_2 = k_3 = 4$. Then the result follows
  numerically:
  \[
    \frac{2n}{8} + 6 \cdot \ceil*{\frac{n}{48} - 1} 
    = 32 + 6 \cdot 2
    = 44
    > 42
    = \ceil*{\frac{n}{3} - 1}
    .
  \]
  Assume from now that either $d \geq 4$, or at least one of $k_2 \geq 5$ or $k_3 \geq 4$. Now
  \begin{align*}
    \frac{2n}{8} + 6 \cdot \ceil*{\frac{n}{48} - 1} 
    \geq \frac{n}{4} + \frac{n}{8} - 6 
    % = \frac{n}{4} + \frac{n}{8} - 6 
    = \frac{n}{3} + \frac{n}{24} - 6 
    .
  \end{align*}
  If $d \geq 4$, then we are done, as
  \begin{align*}
    \frac{n}{3} + \frac{n}{24} - 6 
    \geq \frac{n}{3} + \frac{4^4}{24} - 6 
    > \frac{n}{3}  
    > \ceil*{\frac{n}{3} - 1}
    .
  \end{align*}
  If instead $d = 3$ and either $k_2 \geq 5$ or $k_3 \geq 5$ we are done too, as
  \begin{align*}
    \frac{n}{3} + \frac{n}{24} - 6 
    \geq \frac{n}{3} + \frac{8 \cdot 5 \cdot 4}{24} - 6 
    > \frac{n}{3}  
    > \ceil*{\frac{n}{3} - 1}
    .
  \end{align*}
\end{proof}

\end{document}